\documentclass[a4paper]{llncs}
\usepackage[ruled,linesnumbered,lined,boxed,procnumbered]{algorithm2e}
\usepackage{tikz,float,comment,amsmath,amsfonts,amssymb}
\newtheorem{observation}{Observation}

\definecolor{lightergray}{gray}{1}
\SetCommentSty{mycommfont}

\newlength{\LPlhbox}
\author{S.M.Dhannya \and N.S.Narayanaswamy}
\institute{Department of Computer Science and Engineering, \\ IIT Madras, Chennai, India.}

\begin{document}
\title{Conflict-Free Colouring using Maximum Independent Set and Minimum Colouring}

\maketitle


\begin{abstract}
Given a hypergraph $H$, 
the conflict-free colouring problem is to colour vertices of $H$ using minimum colours so that each hyperedge in $H$ sees a unique colour.
We present a polynomial time reduction from the conflict-free colouring problem in hypergraphs to the maximum independent set problem in a class of simple graphs, which we refer to as \textit{conflict graphs}.  
We also present another characterization of the conflict-free colouring number in terms of the chromatic number of graphs in an associated family of simple graphs, which we refer to as \textit{co-occurrence graphs}. 
We present perfectness results for co-occurrence graphs and a special case of conflict graphs. Based on these results and a linear program that returns an integer solution in polynomial time, we obtain a polynomial time algorithm to compute a minimum conflict-free colouring of interval hypergraphs, thus solving an open problem due to Cheilaris et al.\cite{CPLGARSS2014}.
Finally, we use the co-occurrence graph characterization to prove that for an interval hypergraph, the conflict-free colouring number is the minimum partition of its intervals into sets such that each set has an exact hitting set (a hitting set in which each interval is hit exactly once).
\end{abstract}

%
\section{Introduction}\label{sec:intro}
A colouring of the vertices of a hypergraph $H = (\mathcal{V},\mathcal{E})$ is called  \textit{conflict-free} if every hyperedge $e$ has at least one vertex that has a colour different from other vertices in $e$. Following Smorodinsky et al. \cite{CPLGARSS2014}, we abbreviate conflict-free to \textit{CF} in this paper. 
 The minimum number of colours with which a hypergraph $H$ can be CF coloured is the CF colouring number of the hypergraph and it is denoted by $\chi_{cf}(H)$. Computing the CF colouring number of a given hypergraph and a corresponding colouring is the CF colouring problem. Motivated by a frequency assignment problem in cellular networks, Even, Lotker, Ron and Smorodinsky \cite{ELRS2003} introduced this problem and published the first paper on CF colouring. In mobile communication networks, one must assign frequencies to base stations, such that every client that comes under the transmission range of multiple base stations, can associate itself to a unique base station without any interference from another base station. We can view the transmission range of various base stations as geometric regions in the 2-dimensional plane. CF colouring problem also finds applications in other areas like RFID (Radio Frequency Identification) networks, robotics and computational geometry (See the survey by Somorodinsky\cite{Sm2013}).
The survey due to Smorodinsky \cite{Sm2013}  presents a general framework for CF colouring a hypergraph. 
They showed that if for every induced sub-hypergraph $H' \subseteq H$,  the chromatic number of $H'$ is at most $k$, then $\chi_{cf}(H) \leq \log_{1+ \frac{1}{k-1}} n = O(k \log n)$, where $n = |\mathcal{V}|$. Pach and Tardos \cite{PJGT2009} have shown that if $|\mathcal{E}(H)| <$ $ \binom{s}{2}$ for some positive integer $s$, and $\Delta$ is the maximum degree of a vertex in $H$, then $\chi_{cf}(H)  < s$ and $\chi_{cf}(H) \leq \Delta + 1$. Since its inception, the CF colouring problem has been studied on different types of hypergraphs. Even et al. \cite{ELRS2003} have studied a number of hypergraphs induced by geometric regions on the plane including discs, axis-parallel rectangles, regular hexagons, and general congruent centrally symmetric convex regions in the plane. Let $\mathcal{D}$ be a set of $n$ finite discs in $\mathbb{R}^2$. 
For a point $p \in \mathbb{R}^2$, define $r(p) = \{D \in \mathcal{D} : p \in D\}$. The hypergraph $(\mathcal{D}, \{r(p)\}_{p \in D})$, denoted by $H(\mathcal{D})$, is called the hypergraph induced by $\mathcal{D}$. Smorodinsky showed that $ \chi_{cf}(H(\mathcal{D})) \leq \log_{4/3} n$ \cite{Sm2007}. Similarly, if $\mathcal{R}$ is a set of $n$ axis-parallel rectangles in the plane, then, $\chi_{cf}(H(\mathcal{R})) = O(\log^2n)$. There have been many studies on hypergraphs induced by neighbourhoods in simple graphs. Given a simple graph $G = (V,E)$, the \textit{open neighbourhood} (or simply neighbourhood) of a vertex $v \in V$ is defined as follows: $N(v) = \{u \in V | (u,v) \in E\}$. The set $N(v) \cup v$ is known as the \textit{closed neighbourhood} of $v$. Pach and Tardos \cite{PJGT2009} have shown that the vertices of a graph $G$ with maximum degree $\Delta$ can be coloured with $O(\log^{2+\epsilon} \Delta)$ colours, so that the closed neighbourhood of every vertex in $G$ is CF coloured. They also showed that if the minimum degree of vertices in $G$ is $\Omega(\log \Delta)$, then the open neighbourhood can be CF coloured with at most $O(\log^2 \Delta)$ colours. Abel et al. \cite{Abel2017} gave the following tight worst-case bound for neighbourhoods in planar graphs: three colours are sometimes necessary and always sufficient. Keller and Smorodinsky \cite{Keller2018} studied conflict-colourings of intersection graphs of geometric objects. They showed that the intersection graph of $n$ pseudo-discs in the plane admits a CF colouring with $O(\log n)$ colours, with respect to both closed and open neighbourhoods. Ashok et al. \cite{MAXCFC2015} studied an optimization variant of the CF colouring problem, namely \textsc{Max}-CFC. Given a hypergraph $H = (\mathcal{V},\mathcal{E})$ and integer $r \geq 2$, the problem is to find a maximum-sized subfamily of hyperedges that can be CF coloured with $r$ colours. They have given an exact algorithm running in $O(2^{n+m})$ time. The paper also studies the problem in the parametrized setting where one must find if there exists a subfamily of at least $k$ hyperedges that can be CF coloured using $r$ colours. They showed that the problem is FPT and gave an algorithm with running time $2^{O(k \log \log k + k \log r)} (n+m)^{O(1)}$.

Another line of results in the literature focus on  discrete interval hypergraphs (refer Section \ref{sec:Prelims} for definition)\cite{ELRS2003}. It was shown that a hypergraph formed by all intervals on  a set of $n$ points on a line can be CF coloured using $\Theta(\log n)$ colours. Chen et al.\cite{CFKLMMPSSWW2006} presented results on an online variant of this problem, where a point has to be assigned a colour upon its arrival and the resulting colouring should be conflict-free with respect to all intervals. They gave a greedy algorithm that uses $\Omega(\sqrt n)$ colours, a deterministic algorithm that uses $\Theta(\log^2n)$ colours and a randomized algorithm that uses $O(\log n)$ colours. The case when interval hypergraphs have a given set of intervals as the hyperedges has been of interest \cite{CPLGARSS2014,KATZ}.  Katz et al.\cite{KATZ} gave a polynomial time algorithm for CF colouring an interval hypergraph with approximation ratio 4. Cheilaris et al.\cite{CPLGARSS2014} improved this result in their paper on \textit{$k$-strong CF colouring} problem. The authors gave a polynomial-time approximation algorithm with approximation ratio 2 for $k = 1$ and $5 - \frac{2}{k}$, when $k \geq 2$. Further, they presented a quasipolynomial time algorithm for the decision version of the $k$-strong CF colouring problem. 

\subsection{Preliminaries}\label{sec:Prelims}
We use $deg(v)$ to denote the degree of a vertex $v$. Other definition and notations are from West \cite{West} and Smorodinsky \cite{Sm2013}.
A hypergraph $H_n=([n],\mathcal{I}_n)$, where $[n] = \{1, \ldots, n\}$  and $\mathcal{I}_n=\big\{ \{i, i+1, \ldots, ,j\} \mid i \leq j \text{ and } i,j \in [n] \big\}$ is known as a \textit{discrete interval hypergraph} \cite{CPLGARSS2014}. 
Naturally, a hyperedge in $\mathcal{I}_n$ is called an interval.
A hypergraph such that the set of hyperedges is a family of intervals $\mathcal{I} \subseteq \mathcal{I}_n$ is known as an \textit{interval hypergraph}. In an interval $I = \{i,i+1,\ldots,j\}$, $i$ and $j$ are the \textit{left} and \textit{right endpoints} of $I$ respectively, denoted by $l(I)$ and $r(I)$, respectively.
Since an interval is a finite set of consecutive integers, it follows that $|I|$ is well-defined. Throughout the paper, we assume that the hypergraph $H$ has $n$ vertices and $m$ hyperedges.
 
\noindent
Let $H = (\mathcal{V}, \mathcal{E})$ be a hypergraph and $C : \mathcal{V} \rightarrow \{0, 1, 2, \ldots, k\}$ be a function, which we refer to as a {\em vertex colouring function}. $C$ is a CF colouring \cite{CPLGARSS2014,CS2012} of $H$ using $k$ colours if for every hyperedge $e \in \mathcal{E}$ there exists a non-zero colour $j \in [k]$ such that $|e \cap C^{-1}(j)| = 1$. 
If vertex $v \in e$ has been assigned a colour different from the colour of all other vertices in $e$, then we say that $e$ {\em is CF coloured} by $v$. Specifically, note that we use colour $0$ to indicate that a vertex with colour $0$ does not conflict-free colour any hyperedge.

\begin{definition}
\texttt{CFCIntervals}: Given an interval hypergraph $H = (\mathcal{V},\mathcal{I})$, find a CF colouring of $H$ using minimum number of colours.
\end{definition}

\noindent
A \textit{hitting set} of $H =  (\mathcal{V},\mathcal{E})$ is a set $T \subseteq \mathcal{V}$ that has at least one vertex from every hyperedge. If $T$ intersects every hyperedge exactly once, then $T$ is called an \textit{exact hitting set}. We refer to a hypergraph which has an exact hitting set as an \textit{exactly hittable hypergraph}. The following observation shows that a CF colouring using one non-zero colour is equivalent to an exact hitting set. 
\begin{observation} \label{lem:1CFimpliesEHS}
A hypergraph $H = (\mathcal{V},\mathcal{E})$ has a CF colouring using one non-zero colour if and only if $H$ is exactly hittable.
\end{observation}
\begin{proof}
 If $H$ is an exactly hittable hypergraph, then we get a natural CF colouring with one non-zero colour by giving the non-zero colour to the vertices in the exact hitting set, and the colour 0 to all other vertices. Similarly, in a CF colouring using one colour, each hyperedge contains exactly one vertex which is assigned the non-zero colour, and these vertices form an exact hitting set. Hence the lemma. \qed
\end{proof}
While the recognition of exactly hittable hypergraphs is well-known to be hard, there are polynomial time recognition algorithms for the special case of exactly hittable interval hypergraphs \cite{Dom2006,NDR2018}.

\begin{theorem}[Theorem 4 in \cite{NDR2018}] \label{thm:IHisEH}
There exists a polynomial time algorithm which decides if an interval hypergraph is  exactly hittable.
\end{theorem}
\textbf{Perfect Graphs: }A simple graph $G$ is \textit{perfect} if the chromatic number,  denoted by $\chi(G')$,
 of every induced subgraph $G'$ of $G$ equals the clique number, denoted by $\omega(G')$, of $G'$. A \textit{berge graph} is a simple graph that has neither an odd hole nor an odd anti-hole as an induced subgraph \cite{Berge1985,Chudnovsky2005,Chudnovsky2006,Gol2004}. Recall that an odd hole is an induced cycle of length at least 5 and an odd anti-hole is the complement of an odd hole. 
\begin{theorem}[Strong Perfect Graph Theorem (Theorem 1.2 in \cite{Chudnovsky2006})] \label{thm:SPGT}
A graph is perfect if and only if it is Berge. 
\end{theorem}

The \textit{independence number} of a simple graph $G$ is the size of a maximum independent set of $G$. We denote the independence number of graph $G$ by $\alpha(G)$.

\subsection{Our Results}
\noindent
Our main result is a polynomial time algorithm for the problem of \texttt{CFCIntervals} (defined in Preliminaries). We construct from a given hypergraph two kinds of simple graphs, namely \textit{conflict graphs} in Section \ref{sec:CFCviaIndepGraphs} and \textit{co-occurrence graphs} in Section \ref{sec:cfccog}. In the case of conflict graphs, we give a reduction from the CF colouring problem in hypergraphs to the maximum independent set problem in conflict graphs. In the case of co-occurrence graphs, we present a relation between CF colouring number of $H$ and the chromatic number of its co-occurrence graphs. In general, the maximum independent set problem and the proper colouring problem are NP-hard in simple graphs. However, we prove important structural properties of these graphs when the underlying hypergraphs are interval hypergraphs. We use these properties to eventually arrive at an efficient solution for \texttt{CFCIntervals}. 

First, we present a reduction from CF colouring problem in hypergraphs to the maximum independent set problem in \textit{conflict graphs}. The definition of conflict graphs and the proof of this reduction are given in Section \ref{sec:CFCviaIndepGraphs}. For a hypergraph $H = (\mathcal{V},\mathcal{E})$, and $0 < k \leq n$, the associated conflict graph denoted by $G_k(H)$, has the following relationship with $H$.
\begin{theorem} \label{thm:IndepEqCFC}
Let $H$ be a hypergraph with $m$ hyperedges and $k$ be a positive integer, and  $G_k(H)$ denote the conflict graph of $H$. Let $k_{min}$ be the smallest $k$ for which the independence number of $G_k(H)$ is $m$.  Then, $\chi_{cf}(H) = k_{min}$.
Further, $G_k(H)$ can be constructed in polynomial time and thus the CF colouring problem on hypergraphs can be reduced in polynomial time to the maximum independent set problem in conflict graphs. 
\end{theorem}

\noindent
Next, we present a characterization of the CF colouring number of a hypergraph $H$ in terms of the chromatic number of co-occurrence graphs of $H$. We prove this characterization in Section \ref{sec:cfccog}.  

\begin{theorem} \label{thm:Co-occChar}
Let $H = (\mathcal{V},\mathcal{E})$ be a hypergraph.  
Let $\chi_{min}(H)$ be the minimum chromatic number over all possible co-occurrence graphs of $H$. Then, $ \chi_{cf}(H) = \chi_{min}(H)$.
\end{theorem}

\noindent
In order to use the above characterization to obtain an efficient solution for \texttt{CFCIntervals}, there are two difficulties:
\begin{enumerate}
\item Finding the chromatic number of a simple graph is known to be $NP$-complete \cite{Garey1979}.
\item Finding the co-occurrence graph with the smallest chromatic number involves a search among exponentially many co-occurrence graphs.
\end{enumerate}

\noindent
It is now well known due to the result by Gr{\"o}tschel et al. that the chromatic number of a perfect graph can be found in polynomial time \cite{grotschel2012}. In the theorem below, we show that that each co-occurrence graph of an interval hypergraph is perfect. 
\begin{theorem} \label{thm:Co-occPerf}
The co-occurrence graphs of interval hypergraphs are perfect.
\end{theorem}
This theorem has been proved in Section \ref{sec:Co-occConflictGraphsArePerfect}. It follows that the chromatic number of a co-occurrence graph of an interval hypergraph can be obtained in polynomial time. 

To address the second difficulty, we use a Linear Programming (LP) formulation and design a polynomial time separation oracle for the conflict graph when $k=1$.  This is detailed in Section \ref{sec:SepOracle}. The separation oracle uses fact that the conflict graph of a hypergraph for $k=1$ is perfect.  This perfectness property is stated in Theorem \ref{thm:ResConfGraphIsPerfect}. From the solution of the LP, we arrive at a co-occurrence graph with the smallest chromatic number among all co-occurrence graphs of the given hypergraph. Once such a co-occurrence graph is obtained, an optimal CF colouring of the given hypergraph is a direct result from Theorem \ref{thm:Co-occChar}.  

\begin{theorem} \label{thm:ResConfGraphIsPerfect}
Let $H = (\mathcal{V},\mathcal{I})$ be an interval hypergraph such that there are at least 3 pairwise disjoint intervals in $\mathcal{I}$. Then, the conflict graph $G_1(H)$ is perfect.
\end{theorem}
The above theorem is proved in Section \ref{sec:Co-occConflictGraphsArePerfect}.
If the LP returns a fractional solution, then this solution is appropriately rounded to obtain a feasible integer solution. This integer solution corresponds to a subset of nodes of the conflict graph. From this subset of nodes, a co-occurrence graph is constructed. We show that an optimal CF colouring of an interval hypergraph can be obtained from a proper colouring of the resulting co-occurrence graph.
\begin{theorem} \label{thm:CFCIntHypPolyTime}
\texttt{CFCIntervals} can be solved in polynomial time.
\end{theorem}

\noindent
Finally, we study the relationship between CF Colouring problem and the Exact Hitting Set (EHS) problem.   The EHS problem is the NP-hard dual of the Exact Cover problem \cite{Knuth2000}, which is one of the earliest known NP-hard problems. For a hypergraph $H = (\mathcal{V},\mathcal{E})$, we observe that a CF colouring of $H$ using at most $c$ non-zero colours partitions $\mathcal{E}(H)$ into $c$ hypergraphs such that each hypergraph has an exact hitting set.  This is a very natural observation and has been formally stated in  Lemma \ref{lem:CFimpliesEHS}. A simple proof of this observation has been given in the appendix.  The interesting question is whether a hypergraph which can be partitioned into $c$ hypergraphs, each of which has an exact hitting set, can be CF coloured with at most $c$ non-zero colours.  We answer this question affirmatively in two cases: when $c=1$ (in Lemma \ref{lem:1CFimpliesEHS}) and in the case when the hypergraph is an interval hypergraph (in Theorem \ref{thm:EHS-CF}). Our results are based on the characterization of the CF colouring number presented in Theorem \ref{thm:Co-occChar}. The theorem below has been proved in Section \ref{sec:cfeqehs}.
\begin{theorem} \label{thm:EHS-CF}
For an interval hypergraph $H = (\mathcal{V},\mathcal{E})$, there exists a partition of $\mathcal{E}$ into $k$ parts $\{\mathcal{E}_1, \ldots, \mathcal{E}_k\}$ such that for each $1 \leq i \leq k$, $H_i = (\mathcal{V}, \mathcal{E}_i)$ has an exact hitting set if and only if there exists a CF colouring of $H$ with $k$ non-zero colours.
\end{theorem}

\section{CF colouring,  Chromatic Number and Independence Number}
\subsection{CF Colouring via Conflict Graphs}  \label{sec:CFCviaIndepGraphs}
Given a hypergraph $H = (\mathcal{V},\mathcal{E})$, and a positive integer $k$,  we define the conflict graph $G_k(H) = (V,E)$. Wherever $H$ is implied, we use $G_k$ to denote $G_k(H)$. 
$G_k$ is a simple graph that encodes the constraints to be respected by a CF colouring of $H$ with at most $k$ non-zero colours.  The vertex set of $G_k$ is $V = \big\{(e,v,c) \mid e \in \mathcal{E}, v \in e, 1 \leq c \leq k \big\}$.   The elements of $V$ are referred to as {\em nodes} and the word \textit{vertex} refers to a vertex of a hypergraph.
 In a node $(e,v,c)$,  we refer to $e$ as the hyperedge coordinate, $v$ as the vertex coordinate and $c$ as the colour coordinate. 
Conceptually, a node $(e,v,c)$ in $G_k$ represents the  logical proposition that hyperedge $e$ is CF coloured by vertex $v \in e$ which has been assigned the colour $c$.  
The edge set of $G_k$ is a subset of pairs of $V$.  $E(G_k)$ is defined such that each edge encodes a constraint to be respected by any CF colouring of $H$. The edge set of $G_k$ is $E = E_{vertex} \cup E_{edge} \cup E_{colour}$, where $E_{vertex}, E_{edge}, $ and $E_{colour}$ are defined as follows:
\begin{enumerate}
\item $E_{vertex} = \big\{ \big((e,v,c),(g,v,d)\big) \mid  1 \leq c \neq d \leq k \big\}$.  Note that the set of nodes in $G_k$ whose vertex coordinate is $v$ forms a complete $k$-partite graph.
\item $E_{edge} = \big\{ \big((e,v,c),(e,u,d)\big) \mid 1 \leq c,d \leq k \big\}$. For each hyperedge $e$ in $H$, the nodes in $G_k$ with $e$ as the hyperedge coordinate form a clique.  
\item $E_{colour} = \big\{ \big((e,v,c),(g,u,c)\big) \mid \{v,u\} \subseteq e \text{ or } \{v,u\} \subseteq g, u \neq v, 1 \leq c \leq k \big\}$.  For each colour $c$, an edge $((e,v,c), (g,u,c))$ is created for those $u \neq v$ such that either both $u$ and $v$ are in $e$ or both are in $g$.
\end{enumerate}

There is a natural correspondence between the set of independent sets of $G_k$ and the set of  vertex colourings of $\mathcal{V}(H)$ with $k$ colours. \\

\noindent
\textbf{Independent sets in $G_k$ and vertex colourings of $\mathcal{V}(H)$:}
Given an independent set $\mathcal{A}$ of $G_k$, consider the following vertex colouring function $f: \mathcal{V}\rightarrow \{0,1,2,\ldots,k\}$ defined as follows. For each $v \in \mathcal{V}$,
\[   
f(v) = 
     \begin{cases}
       c, \text{ if }\exists e \text{ such that } (e,v,c) \in \mathcal{A}\\
       0, \text{ otherwise} \ 
     \end{cases}
\]
 Next, given  a vertex colouring function $f'$ of  $\mathcal{V}(H)$, we define a subset of nodes in $G_k$ as follows.  
  Consider the set of hyperedges in $\mathcal{E}(H)$, denoted by $\mathcal{C}$, that are CF coloured by $f'$.  For each hyperedge $e \in \mathcal{C}$, let  $(e,v, c)$ be 
an arbitrary node such that $e$ is CF coloured by $v$  and $f'(v) = c$. Let $\mathcal{A}$ denote the set of these nodes.   
We call $\mathcal{A}$ the \textit{conflict-free set} obtained from $f'$. 
In the following two lemmas, we formally prove the connection between a CF colouring of $H$ and a maximum independent set in $G_k$ by proving properties of vertex colouring functions and  conflict-free sets.

\noindent
{\bf Notation:} In the rest of this section $G_k$ refers to a conflict graph and the corresponding $H$ and $k$ will be clear from the context.
Similarly, the conflict graph associated with a vertex colouring function $f$ obtained from an independent set and the vertex colouring function associated with a conflict-free set in the hypergraph will also be clear from the context.
\begin{lemma} \label{lem:MaxIndIsCFC}
For a positive integer $k$ and a hypergraph $H$, let $G_k$ be the conflict graph of $H$. Suppose $\alpha(G_k) = m$. Then the vertex colouring function $f$ obtained from any maximum independent set in $G_k$ is a CF colouring of $H$ with at most $k$ non-zero colours.
\end{lemma}
\begin{proof}
Let $\mathcal{A}$ be a maximum independent set of $G_k$. For a hyperedge $e$, the nodes of $G_k$ with $e$ as the hyperedge coordinate form a clique.  Therefore, all the $m$ nodes in $\mathcal{A}$ have distinct hyperedge coordinates.  Let $(e,v_e, c_e)$ $\in \mathcal{A}$ be the node corresponding to hyperedge $e$.  From definition of $E_{vertex}$, we observe  that the vertex colouring function $f$ defined based on $\mathcal{A}$ is indeed a function on $\mathcal{V}(H)$ such that the range of $f$ is a subset of $\{0, 1, \ldots, k\}$.  This follows from the fact that if $(e, v_e, c_e)$ and $(g, v_g, c_g)$ are two nodes in $\mathcal{A}$ such that $v_e = v_g$, then $c_e = c_g$.  Indeed, if $c_e$ and $c_g$ were distinct colours, then by the definition of $E_{vertex}$, there should have been an edge $\big((e, v_e, c_e), (g, v_g, c_g)\big)$ which does not exist as the two nodes are in the independent set $\mathcal{A}$.  We next show that $f$ is a CF colouring of $H$ : Let $(e,v_e,c_e) \in \mathcal{A}$. Let $u \neq v_e$ be another vertex in $e$.  We now show that $f(u)$ is different from $c_e$: If $f(u)=0$, then clearly $f(u)$ and $f(v_e)$ are different, since $c_e$ is a non-zero value.  In the case when $f(u)$ is non-zero, then by the definition of $f$ there is a node $(g,v_g,c_g)$ in $\mathcal{A}$ such that  the vertex coordinate $v_g$ is the vertex $u$.   Since $u$ and $v_e$ are distinct elements of the hyperedge $e$, and because $(e,v_e,c_e)$ and $(g,v_g,c_g)$ are nodes in the independent set  $\mathcal{A}$ and thus non-adjacent in $G_k$, it follows that $c_g \neq c_e$. The reason is that had they been equal, by the definition of $E_{color}$, $\big((e,v_e,c_e),(g,v_g,c_g)\big)$ would have been an edge in $E_{color}$, which would contradict the fact that $(e,v_e,c_e)$ and $(g,v_g,c_g)$ are nodes in the independent set $\mathcal{A}$. Therefore, $f(u)=c_g$ and $f(v_e) = c_e$ are different. It follows that for any $u \neq v_e$, $f(u) \neq f(v_e)$.  Hence $f$ is a CF colouring of $H$.  Since the range of $f$ has at most $k$ non-zero colours, it follows that $f$ is a CF colouring of $H$ using at most $k$ non-zero colours. 
\qed
\end{proof}

We next show that a conflict-free set in $G_k$ obtained from a CF colouring using at most $k$ non-zero colours is a maximum independent set in $G_k$.

\begin{lemma}\label{lem:CFCIsMaxInd}
Let $f'$ be a CF colouring of $H$ that uses $k$ non-zero colours. Then, a conflict-free set $\mathcal{A}$ obtained from $f'$ is a maximum independent set of the conflict graph $G_k$.
\end{lemma}
\begin{proof}
Since $f'$ is a CF colouring of $H$, it follows from definition of $\mathcal{A}$ that for each $e \in \mathcal{E}(H)$, there is exactly one node in $\mathcal{A}$ for which the hyperedge coordinate is $e$. Therefore, $|\mathcal{A}| = m$.  Let $\mathcal{A} = \{(e, v_e, c_e) \mid e \in \mathcal{E}(H)\}$.  
We now show that $\mathcal{A}$ is an independent set.  Let $(e,v_e,c_e)$ and $(g,v_g,c_g)$ be two nodes in $\mathcal{A}$.  By construction of $\mathcal{A}$, it follows that $e \neq g$.  We consider two cases: when $v_e = v_g$, then $c_e = c_g$ since they are colours given by $C$ to $v_e = v_g$.  From the definitions of  $E_{vertex}$ and $E_{edge}$, it follows that $(e,v_e,c_e)$ and $(g,v_g,c_g)$ are non-adjacent.  When $v_e \neq v_g$, we have two sub-cases here: when $v_e$ and $v_g$ do not both belong to $e$ and do not both belong to $g$, then by definition of $E(G_k)$, it follows that  $(e,v_e,c_e)$ and $(g,v_g,c_g)$ are non-adjacent.  In the sub-case when both $v_e$ and $v_g$ belong to either $e$ or $g$, then since $f'$ is a conflict-colouring of $H$ in which $e$ and $g$ are CF coloured by $v_e$ and $v_g$, respectively, it follows that $c_e \neq c_g$.  
From the definition of $E_{colour}$, it follows that $(e,v_e,c_e)$ and $(g,v_g,c_g)$ are non-adjacent.  Therefore, $\mathcal{A}$ is an independent set in $G_k$.  Since the nodes of $G_k$ are partitioned into at most $m$ cliques (one clique for each hyperedge $e \in \mathcal{E}(H)$), it follows that the maximum independent set in $G_k$ has at most $m$ nodes. Therefore, $\mathcal{A}$ is a maximum independent set of size $m$ in $G_k$.  \qed
\end{proof}

\begin{proof}[of Theorem \ref{thm:IndepEqCFC}]
The proof of this theorem follows from Lemmas \ref{lem:MaxIndIsCFC} and \ref{lem:CFCIsMaxInd}. Let $\mathcal{A}$ be a maximum independent set of $G_{k_{min}}$. Given $|\mathcal{A}| = m$. Let $f$ be the vertex colouring function obtained from $\mathcal{A}$. It follows from Lemma \ref{lem:MaxIndIsCFC} that $f$ is a CF colouring of $H$ with at most $k_{min}$ colours. Hence $\chi_{cf}(H) \leq k_{min}$. \\
To show the other direction, consider any optimal CF colouring $C_{opt}$ of $H$. Let $k_{opt}$ be the number of colours used by $C_{opt}$. That is, $\chi_{cf}(H) = k_{opt}$. Let $\mathcal{A}$ be the conflict-free set obtained from $C_{opt}$. It follows from Lemma \ref{lem:CFCIsMaxInd} that $\mathcal{A}$ is a maximum independent set of graph $G_{k_{opt}}$ and it is of size $m$.   Therefore, $k_{min} \leq k_{opt} = \chi_{cf}(H)$.  Therefore, it follows that $\chi_{cf}(H) = k_{min}$. \\
Also for each hypergraph $H$ and $k > 0$, it follows from the description of $G_k$ that it can be constructed in polynomial time.  Further, since $\chi_{cf}(H) = k_{min}$ it follows that we have a polynomial time reduction from the CF colouring problem in hypergraphs  to the maximum independent set problem in conflict graphs. 
\qed
\end{proof}

\subsection{Co-occurrence Graphs and Proper Colouring} \label{sec:cfccog}
Given a hypergraph $H=(\mathcal{V}, \mathcal{E})$, we show that the CF colouring number of $H$ is the minimum chromatic number over a set of simple graphs called \textit{co-occurrence graphs}.  A co-occurrence graph of $H$ is defined based on a \textit{representative function} $t$ (defined below), and is denoted by $\Gamma_t$.   Note that for a conflict graph $G_k(H)$, the subscript is a positive integer, whereas in the case of the co-occurrence graph $\Gamma_t$, the subscript $t$ is a representative function.
For a CF colouring function $C$ defined on $\mathcal{V}$,
let $t:\mathcal{E} \rightarrow \mathcal{V}$ be a function such that $e \in \mathcal{E}$ is CF coloured by $t(e)$. We refer to $t$ as a representative function obtained from the colouring $C$.  Further, given a function $t:\mathcal{E} \rightarrow \mathcal{V}$ such that for each edge $e$, $t(e) \in e$, we define a CF colouring of $H$ for which $t$ is the representative function as follows. 
Let $R \subseteq \mathcal{V}$ denote the image of $\mathcal{E}$ under the function $t$. The vertex set of the co-occurrence graph $\Gamma_t$ is  $R$, and for $u,v \in R$,  $(u,v)$ is an edge in  $\Gamma_t$ if and only if  for some $e \in \mathcal{E}$, $u \in e$ and $v \in e$ and $t(e)$ is either $u$ or $v$.  
An example is given in Figure \ref{fig:coOccur}. 
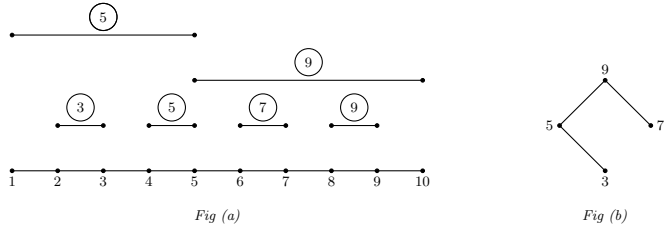
\begin{figure}[ht]
\centering
\begin{tikzpicture}[scale=0.6, every node/.style={scale=0.6}]
\draw(1,3)--(5,3); 
\draw(5,2)--(10,2); 
\draw(2,1)--(3,1); 
\draw(4,1)--(5,1); 
\draw(6,1)--(7,1); 
\draw(8,1)--(9,1); 

\draw[radius=3mm, color=black](3,3.4) circle;
\draw(3,3.4) node{$5$};
\draw[radius=3mm, color=black](7.5,2.4) circle;
\draw(7.5,2.4) node{$9$};
\draw[radius=3mm, color=black](2.5,1.4) circle;
\draw(2.5,1.4) node{$3$};
\draw[radius=3mm, color=black](4.5,1.4) circle;
\draw(4.5,1.4) node{$5$};
\draw[radius=3mm, color=black](6.5,1.4) circle;
\draw(6.5,1.4) node{$7$};
\draw[radius=3mm, color=black](8.5,1.4) circle;
\draw(8.5,1.4) node{$9$};

\draw[radius=0.4mm, color=black, fill=black](1,0) circle node[below=0.02cm]{1};
\draw[radius=0.4mm, color=black, fill=black](2,0) circle node[below=0.02cm]{2};
\draw[radius=0.4mm, color=black, fill=black](3,0) circle node[below=0.02cm]{3};
\draw[radius=0.4mm, color=black, fill=black](4,0) circle node[below=0.02cm]{4};
\draw[radius=0.4mm, color=black, fill=black](5,0) circle node[below=0.02cm] {5};
\draw[radius=0.4mm, color=black, fill=black](6,0) circle node[below=0.02cm]{6}; 
\draw[radius=0.4mm, color=black, fill=black](7,0) circle node[below=0.02cm]{7};
\draw[radius=0.4mm, color=black, fill=black](8,0) circle node[below=0.02cm]{8};
\draw[radius=0.4mm, color=black, fill=black](9,0) circle node[below=0.02cm]{9};
\draw[radius=0.4mm, color=black, fill=black](10,0) circle node[below=0.02cm]{10};

\draw[radius=0.4mm, color=black, fill=black](1,3) circle node[below=0.02cm]{};
\draw[radius=0.4mm, color=black, fill=black](5,3) circle node[below=0.02cm]{};
\draw[radius=0.4mm, color=black, fill=black](5,2) circle node[below=0.02cm]{};
\draw[radius=0.4mm, color=black, fill=black](10,2) circle node[below=0.02cm]{};
\draw[radius=0.4mm, color=black, fill=black](2,1) circle node[below=0.02cm]{};
\draw[radius=0.4mm, color=black, fill=black](3,1) circle node[below=0.02cm]{};
\draw[radius=0.4mm, color=black, fill=black](4,1) circle node[below=0.02cm]{};
\draw[radius=0.4mm, color=black, fill=black](5,1) circle node[below=0.02cm]{};
\draw[radius=0.4mm, color=black, fill=black](6,1) circle node[below=0.02cm]{};
\draw[radius=0.4mm, color=black, fill=black](7,1) circle node[below=0.02cm]{};
\draw[radius=0.4mm, color=black, fill=black](8,1) circle node[below=0.02cm]{};
\draw[radius=0.4mm, color=black, fill=black](9,1) circle node[below=0.02cm]{};

\draw[radius=3mm, color=black](3,3.4) circle;

\draw(1,0)--(10,0); 
\draw(5.5,-1) node{\textit{Fig (a)}};
\draw[radius=0.4mm, color=black, fill=black](14,2) circle node[above=0.02cm]{9};
\draw[radius=0.4mm, color=black, fill=black](13,1) circle node[left=0.02cm]{5};
\draw[radius=0.4mm, color=black, fill=black](15,1) circle node[right=0.02cm]{7};
\draw[radius=0.4mm, color=black, fill=black](14,0) circle node[below=0.02cm]{3};

\draw(14,2)--(13,1);
\draw(13,1)--(14,0);
\draw(14,2)--(15,1);
\draw(14,-1) node{\textit{Fig (b)}};
\end{tikzpicture}
\caption{(a) Interval Hypergraph $H = (\mathcal{V},\mathcal{E})$ (b) Co-occurrence graph $\Gamma_t$ of $H$ with $R=\{3,5,7,9\}$, and $t(E)$ for each $E \in \mathcal{E}$ marked as the label for each interval}
\label{fig:coOccur}
\end{figure}%
\noindent
Define $\chi_{min}(H) = \min\limits_{t} \chi(\Gamma_t)$ where $\chi(\Gamma_t)$ is the chromatic number of the co-occurrence graph $\Gamma_t$ and the minimum is taken over all representative functions $t$.  
\begin{proof}[of Theorem \ref{thm:Co-occChar}]
Let $t$ be a representative function such that $\chi(\Gamma_t) = \chi_{min}(H)$.  We extend a proper colouring $C$ of $\Gamma_t$ to a vertex colouring function $C'$ of $\mathcal{V}(H)$ by assigning the colour $0$ to those vertices in $\mathcal{V}(H) \setminus R$.    $C'$ is a CF colouring of $H$ since for each $e \in \mathcal{E}$, the colour assigned to the vertex $t(e)$ by $C'$ is different from the colour assigned to every other vertex in $e$.  The reason for this is as follows: let  $v \in e$ be  a vertex different from $t(e)$. If $C'(v)=0$, then definitely its colour is different from $C'(t(e))$. On the other hand,  if $C'(v)$ is non-zero, then it implies that there is an $e'$ such that $v = t(e')$. Consequently, $v \in V(\Gamma_t)$, and since $v \in e$, $(v,t(e))$ is an edge in $\Gamma_t$ by the definition of $\Gamma_t$.  Further, since $C'$ is obtained from a proper colouring $C$ of $\Gamma_t$ it follows that $C'(v)$ is different from $C'(t(e))$. Thus $\chi_{cf} \leq \chi_{min}(H)$.  We prove that
$\chi_{min}(H) \leq \chi_{cf}(H)$ as follows: since a minimum CF colouring of $H$ gives a representative function $t$ as defined above, it follows that $\chi_{cf}(H) \geq \chi(\Gamma_t) \geq \chi_{min}(H)$.  Therefore, it follows that $\chi_{cf}(H) = \chi_{min}(H)$.  
\qed
\end{proof}
\subsection{Intervals: $G_1$ and Co-occurrence graphs are Perfect } \label{sec:Co-occConflictGraphsArePerfect}

In this section, we prove two perfectness results when the underlying hypergraph is an interval hypergraph. The perfectness of co-occurrence graphs, given in Theorem \ref{thm:Co-occPerf} enables us to find a proper colouring of $\Gamma_t$ in polynomial time. The second perfectness result in Theorem \ref{thm:ResConfGraphIsPerfect} is used to prove Lemma \ref{lem:SepOracPoly}. We first prove Theorem \ref{thm:Co-occPerf}.

\begin{proof}[of Theorem \ref{thm:Co-occPerf}]
We use Theorem \ref{thm:SPGT} to prove this perfectness result. Given an interval hypergraph $H$, let $t$ be a representative function and let $\Gamma_t$ be the resulting co-occurrence graph. We first show that $\Gamma_t$ does not have an induced cycle of length at least 5.  Note that we prove a stronger statement than required by Theorem \ref{thm:SPGT} which requires that there are no induced odd cycles of length at least 5.  Our proof is by contradiction. Assume that $F = \{p_1,p_2 \ldots p_r\}$ is an induced $C_r$-cycle for $r \geq 5$. Let the sequence of nodes in $F$ be $p_1,p_2 \ldots p_r,p_1$. Let $p_i$ be the rightmost point of $F$ on the line.  In what follows, the arithmetic among the indices of $p$ is $mod$ $r$.  
Without loss of generality, let us assume that $p_{i-1} < p_{i+1}$, which are the two neighbours of $p_i$ in $F$.  Therefore, $p_{i-1} < p_{i+1} < p_i$.  Since edge $(p_{i-1},p_i)$ is in $F$, it follows that there exists an interval $I$ for which $t(I) \in \{p_{i-1},p_i\}$. We claim that $t(I)$ is $p_i$: if $t(I)$ is $p_{i-1}$, then $(p_{i-1},p_{i+1})$ is an edge in $\Gamma_t$ by definition.
Therefore, $(p_{i-1},p_{i+1})$ is a chord in $F$, a contradiction to the fact that $F$ is an induced cycle.  Therefore, $t(I) = p_i$.  Further, we claim that the point $p_{i+2} < p_{i-1}$: if $p_{i+2} > p_{i-1}$, then $p_{i+2}$ belongs to the interval $I$ and by the definition of the edges in $\Gamma_t$, $(p_i,p_{i+2})$ is an edge in $\Gamma_t$.  Therefore,  $(p_i,p_{i+2})$ is a chord in $F$. This contradicts the fact that $F$ is an induced cycle. Therefore, $p_{i+2} < p_{i-1}$.   At this point in the proof we have concluded that $p_{i+2} < p_{i-1} < p_{i+1} < p_i$ and $t(I) = p_i$.  Since $(p_{i+1},p_{i+2})$ is an edge in $F$, it follows that there exists an interval $J$ such that both $p_{i+1}$ and $p_{i+2}$ belong to $J$ and $t(J) \in \{p_{i+1},p_{i+2}\}$.  
Since $F$ is an induced cycle of length at least 5, $(p_{i-1},t(J))$ is an edge in $\Gamma_t$ by definition. Therefore,  $(p_{i-1},t(J))$ is a chord in either case,  that is when $t(J) = p_{i+1}$ or $t(J) = p_{i+2}$.  This contradicts the assumption that $F$ is an induced cycle of length at least 5.  Thus, $\Gamma_t$ cannot have an induced cycle of size at least 5.
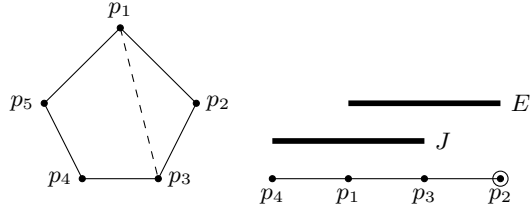
\begin{figure}[ht]
\centering
\begin{tikzpicture}
\draw[radius=0.4mm, color=black, fill=black](4,2) circle node[above=0.02cm]{$p_1$};
\draw[radius=0.4mm, color=black, fill=black](5,1) circle node[right=0.02cm]{$p_2$};
\draw[radius=0.4mm, color=black, fill=black](4.5,0) circle node[right=0.02cm]{$p_3$};
\draw[radius=0.4mm, color=black, fill=black](3.5,0) circle node[left=0.02cm]{$p_4$};
\draw[radius=0.4mm, color=black, fill=black](3,1) circle node[left=0.02cm]{$p_5$};

\draw(4,2)--(5,1);
\draw(5,1)--(4.5,0);
\draw(4.5,0)--(3.5,0);
\draw(3,1)--(3.5,0);
\draw(3,1)--(4,2);
\draw[style=dashed](4,2)--(4.5,0);

\draw[radius=0.4mm, color=black, fill=black](7,0) circle node[below=0.02cm]{$p_1$};
\draw[radius=0.4mm, color=black, fill=black](8,0) circle node[below=0.02cm]{$p_3$};
\draw[radius=0.4mm, color=black, fill=black](9,0) circle node[below=0.02cm]{$p_2$};
\draw(6,0)--(9,0);
\draw[line width=0.75mm](7,1)--(9,1);
\node at (9,1) [right=0.02cm]{$E$};
\draw[radius=1mm](9,0) circle;
\draw[radius=0.4mm, color=black, fill=black](6,0) circle node[below=0.02cm]{$p_4$};
\draw[line width=0.75mm](6,0.5)--(8,0.5);
\node at (8,0.5) [right=0.02cm]{$J$};
\end{tikzpicture}
\caption{Case when $r=5$ and $i=2$}
\end{figure}

\noindent
Next, we show that $\Gamma_t$ does not contain the complement of an induced cycle of length at least 5.
Assume that $F$ is an induced $\overline{C_{r}}$, $r \geq 5$ in $\Gamma_t$. 
\begin{figure}[ht]
\centering
\begin{tikzpicture}[scale=0.8, every node/.style={scale=0.8}]
\draw(4,0)--(4, 1.5);
\draw(4,1.5)--(5.5, 2.5);
\draw(4,0)--(5.5, 2.5);
\draw(5.5,2.5)--(7, 1.5);
\draw(7,1.5)--(7,0);
\draw(5.5,2.5)--(7,0);
\draw(4,1.5)--(5.5,-0.75);
\draw(5.5,-0.75)--(7,1.5);
\draw(4,1.5)--(5.5,1.6);
\draw(5.5,1.6)--(7,1.5);
\draw(4,0)--(5.5,1.6);
\draw(5.5,1.6)--(7,0);
\draw(4,0)--(5.5,-0.75);
\draw(7,0)--(5.5,-0.75);

\draw[radius=0.3mm, color=black, fill=black](4,1.5) circle node[left=0.02cm]{$q_1$};
\draw[radius=0.3mm, color=black, fill=black](4, 0) circle node[below=0.02cm]{$q_2$};
\draw[radius=0.3mm, color=black, fill=black](5.5,2.5) circle node[above=0.02cm]{$q_3$};
\draw[radius=0.3mm, color=black, fill=black](5.5,1.6) circle node[above=0.02cm]{$q_4$};
\draw[radius=0.3mm, color=black, fill=black](5.5,-0.75) circle node[below=0.02cm]{$q_{r-2}$};
\draw[radius=0.3mm, color=black, fill=black](7,1.5) circle node[right=0.02cm]{$q_{r}$};
\draw[radius=0.3mm, color=black, fill=black](7,0) circle node[below=0.02cm]{$q_{r-1}$};

\node at (5.5,-1.5){$(a)$};

\draw[radius=0.2mm, color=black, fill=black](5.5,0.25) circle;
\draw[radius=0.2mm, color=black, fill=black](5.5,0.5) circle;
\draw[radius=0.2mm, color=black, fill=black](5.5,0.75) circle;
\draw[radius=0.2mm, color=black, fill=black](5.5,1) circle;
\draw[radius=0.2mm, color=black, fill=black](5.5,1.25) circle;
\draw[radius=0.2mm, color=black, fill=black](5.5,0) circle;
\draw[radius=0.2mm, color=black, fill=black](5.5,-0.25) circle;

\draw(4,0)--(4.5,0.1);
\draw(4,0)--(4.5,0.2);
\draw(4,0)--(4.5,0.3);

\draw(7,0)--(6.5,0.1);
\draw(7,0)--(6.5,0.2);
\draw(7,0)--(6.5,0.3);

\draw(4,1.5)--(4.5,1.2);
\draw(4,1.5)--(4.5,1.4);
\draw(4,1.5)--(4.5,1.3);

\draw(7,1.5)--(6.5,1.2);
\draw(7,1.5)--(6.5,1.4);
\draw(7,1.5)--(6.5,1.3);

\draw(9,2)--(11,2);
\draw(11,0)--(9,0);
\draw(9,0)--(9,2);
\draw(11,0)--(11,2);

\draw[radius=0.3mm, color=black, fill=black](9,2) circle node[above=0.02cm]{$q_1$};
\draw[radius=0.3mm, color=black, fill=black](11,2) circle node[above=0.02cm]{$q_{r-1}$};
\draw[radius=0.3mm, color=black, fill=black](9,0) circle node[below=0.02cm]{$q_{r}$};
\draw[radius=0.3mm, color=black, fill=black](11,0) circle node[below=0.02cm]{$q_2$};
\node at (10,-1){$(b)$};

\end{tikzpicture}
\caption{(a) Adjacencies of vertices in $\overline{C}$ (b) Induced $C_4$ in complement of $\overline{C}$}%
\label{fig:PerfComple}
\end{figure}
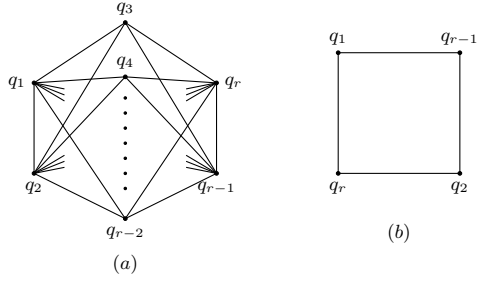%

Let $q_1, q_2, \ldots, q_r$ be the nodes of $F$. Also, let $q_1 < q_2 < \ldots < q_r$ be the left to right ordering of points on the line corresponding to vertices of $F$. Since $deg(q_i) = r-3$ for all $q_i$ in $F$, it follows that no interval $I$, such that $t(I) \in F$, contains more than $r-2$ vertices from $F$. Otherwise, if there exists an interval $I$ such that $t(I) \in F$ contains more than $r-2$ vertices from $F$, then $deg(t(I)) \geq r-2$ in $F$ which is a contradiction. Therefore, there does not exist any interval that contains both $q_1$ and $q_r$. Similarly, there does not exist any interval that contains both $q_1$ and $q_{r-1}$ and any interval that contains both $q_2$ and $q_{r}$. Since $deg(q_1) = r-3$, it follows that $q_1$ must be adjacent to all vertices in $\{q_2,q_3,\ldots,q_{r-2}\}$.
Similarly, $q_r$ must be adjacent to all vertices in $\{q_3,q_4,\ldots,q_{r-1}\}$. Next, we consider the degrees of vertices $q_2$ and $q_{r-1}$ in $F$. Since they are in $F$, $q_2$ is adjacent to $q_1$ and $q_{r-1}$ is adjacent to $q_r$. Now, $q_2$ must be adjacent to $r-4$ more vertices. We show that $q_{2}$ is not adjacent to $q_{r-1}$. Suppose not, that is, if $q_{2}$ is adjacent to $q_{r-1}$, then there exists an interval $I$ that contains both $q_{2}$ and $q_{r-1}$ and $t(I) \in \{q_2, q_{r-1}\}$. Then $t(I)$ is adjacent to all points in the set $\{\{q_2,q_3,\ldots,q_{r-1}\} \setminus t(I)\}$. Thus, by considering the one additional edge incident on $t(I)$ depending on whether $t(I)=q_2$ or $q_{r-1}$, it follows that $deg(t(I)) \geq r-2$, a contradiction to the fact that the degree of each vertex inside $F$ is $r-3$.   Therefore, it follows that  $(q_{2},q_{r-1})$ does not exist in $F$. It follows that in $\overline{F}$, which we know is an induced cycle of length at least 5,   there is an induced cycle $q_1, q_{r-1}, q_2, q_r, q_1$ of length $4$.  This contradicts the structure of an induced cycle of length at least 5.   Hence, we conclude that $\Gamma_t$ does not have an induced cycle of length 5 or more or its complement.  Therefore $\Gamma_t$ is a perfect graph. \qed
\end{proof}%
We now prove the perfectness property of conflict graphs for $k=1$. This property is crucial in our proof of Lemma \ref{lem:SepOracPoly}, where we show that the separation oracle referred by the ellipsoid method runs in polynomial time.
In the proof of the theorem below, $\mu(H)$ denotes the number of vertices in $G_1$. Note that $\mu(H) = \sum_{I \in \mathcal{I}} |I|$. 
\begin{proof}[of Theorem \ref{thm:ResConfGraphIsPerfect}]
By the characterization of perfect graphs in Theorem \ref{thm:SPGT}, we know that for each $p > 1$, induced  odd cycle $C_{2p+1}$ and its complement denoted by $\overline{C_{2p+1}}$ are forbidden induced subgraphs.  We now show that for an interval hypergraph which has at least three disjoint intervals, the graph $G_1$ is perfect.   
Our proof is by starting with the hypothesis that the claim is false and deriving a contradiction.

Let $H = (\mathcal{V},\mathcal{J})$ be an interval hypergraph for which $G_1$ is not perfect, and among all such interval hypergraphs, $H$ minimizes $\mu(H)$. Since $G_1$ is not perfect, let us consider a minimal induced subgraph of $G_1$, denoted by, say $F$ for which $\omega(F) \neq \chi(F)$. 
We claim that for every interval $I \in \mathcal{J}$, both the nodes $(I,l(I))$ and $(I,r(I))$ belong to $F$. The proof of this claim is by contradiction to the fact that $H$ is an interval hypergraph that minimizes $\mu(H)$ and for which $G_1$ is not perfect.  Let $I$ be an interval in $\mathcal{J}$ such that the node $(I,r(I)) \notin V(F)$. Consider the hypergraph $H' = (\mathcal{V},\mathcal{J}')$ where  $\mathcal{J}' = (\mathcal{J} \setminus I) \cup (I \setminus r(I))$.  Let $G_1'$ denote the conflict graph of $H'$. Observe that $V(G_1') = V(G_1) \setminus \{(I, r(I))\}$. Since $(I,r(I)) \notin V(F)$ and $(I,r(I)) \notin V(G_1')$, it follows that $F$ is an induced subgraph of $G_1'$ also.  Therefore by Theorem \ref{thm:SPGT} it follows that $G_1'$ is imperfect.  Further, $\mu(H') < \mu(H)$.   This contradicts the hypothesis that $H$ is the interval hypergraph with minimum $\mu(H)$ for which $G_1$ is imperfect.  Therefore, it follows that for each interval $I \in \mathcal{J}$, $(I,r(I))$ is a node in $F$. An identical argument shows that for each interval $I \in \mathcal{J}$, $(I,l(I))$ is also a node in $F$. Hence it follows that $\forall I \in \mathcal{J}$, both the nodes $(I,l(I))$ and $(I,r(I))$ belong to $F$.  We now consider two exhaustive cases to obtain a contradiction to the known structure of $F$ which we know is either a $C_{2p+1}$ or a $\overline{C_{2p+1}}$ for some $p > 1$.   \\
{\em Case 1- When $F$ is an induced odd cycle $C_{j}, j \geq 5$}: In the preceding argument, we showed that for each interval $I$, both the nodes $(I,l(I))$ and $(I,r(I))$ belong to $F$. This accounts for an even number of distinct nodes in the cycle $C_j$.   Since $C_j$ is an induced odd cycle, it follows that in $C_j$ there is at least one more node $(I,q)$ for which $q$ is different from $r(I)$ and $l(I)$. 
From the definition of $E_{edge}$, we know that the 3 nodes $(I,l(I)), (I,r(I)), (I,q)$ form a $K_3$. This is a contradiction to the fact that an induced cycle of length at least 5 does not have a $K_3$ as an induced subgraph.  Therefore, $F$ is not an induced odd cycle.\\
\noindent
{\em Case 2- When $F$ is the complement of an odd cycle, say $\overline{C_{j}}, j \geq 5$}: Consider 3 pairwise disjoint intervals  $I_1,I_2,I_3$  in $\mathcal{J}$. We have already shown that for each interval $I$, $(I,l(I))$ and $(I,r(I))$ belong to $F$.
It follows that $(I_1,l(I_1))$, $(I_1,r(I_1))$, $(I_2,l(I_2))$, $(I_2,r(I_2))$, $(I_3,l(I_3))$, $(I_3,r(I_3))$ belong to $F=\overline{C_{j}}$. Since $I_1,I_2,I_3$ are pairwise disjoint, it follows from the construction of the conflict graph that for all $1 \leq i \neq j \leq 3$, there is no edge from $(I_i,l(I_i))$ to $(I_j,l(I_j))$ and there is no edge from $(I_i,r(I_i)$ to $(I_j,r(I_j))$. It follows that $(I_1,l(I_1)), (I_2,l(I_2)), (I_3,l(I_3))$ form an independent set. That is, there is an independent set of size at least 3 in $F$. This is a contradiction to the fact that in the complement of any induced cycle of length at least 4, there is no independent set of size greater than 2.   Therefore, $F$ cannot be the complement of an induced odd cycle of length at least 5.\\
Therefore, the assumption of a minimal $H$ for which $G_1$ is not perfect leads to a contradiction to the known structure of graphs which are not perfect. 
 Therefore, our hypothesis that there is a minimal $H$ for which $G_1$ is not perfect is wrong.
Hence, it follows that for an interval hypergraph with at least 3 disjoint intervals, $G_1$ is perfect. 
\qed
\end{proof}
\section{Computing the Optimal Co-occurrence Graph of Interval Hypergraphs using Conflict Graphs}
Throughout this section, we consider conflict graph $G_k(H)$ for $k = 1$. 
Wherever $H$ is implied, we use $G_1$ instead of $G_1(H)$ to denote this special case. Since $k=1$ in $G_1$, the third co-ordinate of every node in $G_1$ is redundant. Hence every node in $G_1$ is a $2$-tuple whose first component is the hyperedge co-ordinate and the second component is the vertex co-ordinate. 

\begin{algorithm}[H]
\caption{CFC-Intervals}
\label{algo:MainAlgo}
\KwIn{Interval hypergraph $H = (\mathcal{V},\mathcal{I})$}
\vspace{3mm}
\If {$\mathcal{I}$ has at least 3 disjoint intervals}{
$B_{opt} = \{\emptyset\}$ \;
$G_1 \leftarrow $ conflict graph of $H$ for $k=1$ \;
$\mathcal{I}' \leftarrow \mathcal{I}$ \;
$q \leftarrow 1$ \;
\While{there is no feasible solution for \texttt{SPAlg}$(\mathcal{B},q)$}{
$q \leftarrow q+1$ \;
}
$q_{min} \leftarrow q$ \;
$B_{opt} \leftarrow $\texttt{SPAlg}$(\mathcal{B},q_{min})$ \;
$B_{optI} \leftarrow $ \texttt{RoundingAlgo}($B_{opt},\mathcal{I}'$) \;
$X^1 \leftarrow \{x_{I,u} \mid x_{I,u} = 1 \text { in solution } B_{optI} \}$ \;
$R \leftarrow \{(I,u) \mid x_{I,u} \in X^1\}$ \; 
Define $\big(t: \mathcal{I} \rightarrow R\big)$ as $t(I) = u \text{ for } (I,u) \in R$ \;
$\Gamma_t \leftarrow$ Co-occurrence graph on $t$ \;
$\chi \leftarrow $ A proper colouring of $\Gamma_t$ \;
\For {each $v \in \mathcal{V}$}{
\If {$v \in R$}{
$\chi_{cf}(v) \leftarrow \chi(v)$ \; 
}
\Else{$\chi_{cf}(v) \leftarrow 0$ \;
}
}}
\Else{
$v_1,v_2 \leftarrow$ points corresponding to a minimum clique cover of $H$ \;
$\chi_{cf}(v_1) \leftarrow 1$ \;
$\chi_{cf}(v_2) \leftarrow 2$ \;
\For {each vertex $v \in \mathcal{V} \setminus \{v_1,v_2\}$}{
$\chi_{cf}(v) \leftarrow 0$ \;
}
}
return $\chi_{cf}$ \;
\end{algorithm}

\begin{observation} \label{obs:NodesOfaVertexAreIndep}
Let $H = (\mathcal{V},\mathcal{I})$ be any hypergraph. For each vertex $v \in \mathcal{V}$, the set of nodes $\{(I,v) \mid I \in \mathcal{I}, v \in I\}$ in $G_1$ forms an independent set.   
\end{observation}

\begin{observation}
The conflict graph $G_1$ has two types of cliques:
\begin{itemize}
\item Set of Type 1 cliques denoted by $\mathcal{Q}_1$ : The set of maximal cliques formed by nodes having the same hyperedge co-ordinate. All edges in this type of clique belong to $E_{edge}$.
\item Set of Type 2 cliques denoted by $\mathcal{Q}_2$ : The maximal cliques in $G_1$ that have at least one edge from $E_{colour}$. 
\end{itemize}
\end{observation}
\noindent
\subsection{Representative Function from a Hitting Set of cliques in $G_1$} 
\label{sec:CFCfromEHSofType1Cliques}
Let $S \subseteq V(G_1)$ be an exact hitting set of $\mathcal{Q}_1$ that hits every maximal clique in $\mathcal{Q}_2$ at most $q$ times, for some integer $q>0$. 
In other words, for each maximal clique $Q \in \mathcal{Q}_1$, $|S \cap Q|= 1$  and for each maximal clique $Q \in \mathcal{Q}_2$, $|S \cap Q| \leq q$. In Section \ref{sec:LP}, we obtain such an exact hitting set by way of a linear program. 
Let $q_{min}$ be the smallest value of $q$ for which such an exact hitting set, say $S_{min}$, exists. Note that $|S_{min}| = m$ since there are $m$ maximal cliques in $\mathcal{Q}_1$, each corresponding to an interval. Let the nodes in $S_{min}$ define a mapping $t: \mathcal{I} \rightarrow \mathcal{V}$ as follows: $t(I) = u \text{ if } (I,u) \in S_{min}$. We show that $t$ is a representative function in Lemma \ref{lem:tIsRepAndOmegaGtLThanKMin}. We also show that the size of the maximum clique in the co-occurrence graph $\Gamma_t$, obtained from the representative function $t$, is upperbounded by $q_{min}$. 
\begin{lemma} \label{lem:tIsRepAndOmegaGtLThanKMin}
Let $t:\mathcal{I} \rightarrow \mathcal{V}$ be the function as defined above.  Then $t$ is a representative function obtained from some conflict-free coloring and $\omega(\Gamma_t) \leq q_{min}$. 
\end{lemma}
\begin{proof}
Since $S_{min}$ is an exact hitting set of maximal cliques of Type 1, it follows that for every interval $I \in \mathcal{I}$, there exists exactly one node in $S_{min}$ whose hyperedge co-ordinate is $I$. Hence, $t$ is indeed a function. Since every interval is assigned a unique representative by $t$, it follows from the proof of Theorem \ref{thm:Co-occChar} that
any proper colouring of $\Gamma_t$ is a CF colouring of $H$.  Therefore,  $t$ is a representative function obtained from such a CF colouring of $H$.

Now, we show that $\omega(\Gamma_t) \leq q_{min}$.  To prove this, we show that $\omega(\Gamma_t)$ is at most the size of the maximum clique in $G_1[S_{min}]$, that is, the induced subgraph of $G_1$ on the node set $S_{min}$.  In particular, we for each clique in $\Gamma_t$ we identify a clique of the same size in $G_1[S_{min}]$.  The proof is by induction on the size of a clique in $\Gamma_t$. The base case is for a clique of size $1$ in $\Gamma_t$.  Clearly, there is a clique of size at least 1 in $G_1[S_{min}]$.   By the induction hypothesis, corresponding to a clique comprising of $u_1,u_2,\ldots,u_{q-1}$ in $\Gamma_t$, there is a maximal clique containing nodes $(I_1,u_1),(I_2,u_2),\ldots,(I_{q-1},u_{q-1})$ in $G_1[S_{min}]$. Now, we prove the claim when there are $q$ vertices in a clique in $\Gamma_t$. Let $u_1,u_2,\ldots,u_q$ be the set of vertices in the clique. Without loss of generality, assume that $u_1,u_2,\ldots,u_{q-1},u_q$ is the left to right ordering of points on the line. Observe that the edge between $u_1$ and $u_q$ exists in $\Gamma_t$ because there exists an interval, say $I'$ such that $u_1$ and $u_q$ belong to $I'$ and $t(I') \in \{u_1,u_q\}$. It follows that the node $(I',t(I'))$ belongs to $S_{min})$. Since both $u_1$ and $u_q$ belong to the interval $I'$, it follows that $u_2,\ldots,u_{q-1}$ belong to interval $I'$.  Therefore,  it follows that $(I',t(I'))$ is adjacent to all nodes $(I_1,u_1),(I_2,u_2),\ldots,(I_{r-1},u_{q-1})$ in $G_1$. 
Hence $(I',t(I'))$ is adjacent to all nodes $(I_1,u_1),(I_2,u_2),\ldots,(I_{r-1},u_{q-1})$ in the induced subgraph $G_1[S_{min}]$. 
It follows that $(I_1,u_1),(I_2,u_2),\ldots,(I_{r-1},u_{q-1})$ and $(I',t(I'))$ form a clique of size $q$ in $G_1[S_{min}]$. Hence the proof.
\qed
\end{proof}
We next show that finding a CF coloring is equivalent to finding an exact hitting set of $\mathcal{Q}_1$ such that cliques in $\mathcal{Q}_2$ are hit as few times as possible.
\begin{lemma}
\label{lem:hitsetequiv}
There exists a set $S \subseteq V(G_1)$ such that for each $Q \in \mathcal{Q}_1$, $|S \cap Q| = 1$ and for each $Q' \in \mathcal{Q}_2$, $|S \cap Q'| \leq q$ if and only if there is a CF colouring of $H$ with $q$ colours.
\end{lemma}
\begin{proof}
Let $S$ be a subset of $V(G_1)$ such that for each $Q \in \mathcal{Q}_1$, $|S \cap Q| = 1$ and for each $Q' \in \mathcal{Q}_2$, $|S \cap Q'| \leq q$. Then by Lemma \ref{lem:tIsRepAndOmegaGtLThanKMin}, there exists a representative function $t$ such that $\omega(\Gamma_t) \leq q$. Since co-occurrence graphs are perfect by Theorem \ref{thm:Co-occPerf}, it follows 
that $\chi(\Gamma_t) = \omega(\Gamma_t)$. It further follows from Theorem \ref{thm:Co-occChar} that a proper colouring of $\Gamma_t$ is a CF colouring of $H$ using $\omega(\Gamma_t) \leq q$ colours.  This completes the forward direction of the claim.

Now, we prove the reverse direction. Let $C$ be a CF colouring of $H$ using $q'$ colours. Then by Theorem \ref{thm:Co-occChar}, $C$ gives a representative function $t'$ with the property $\chi_{cf}(H) \geq \chi(\Gamma_{t'})$. Since $\Gamma_{t'}$ is perfect, we have $\chi_{cf}(H) \geq \omega(\Gamma_{t'})$. It follows that $q' \geq \chi_{cf}(H) \geq \omega(\Gamma_{t'})$. Define $S' \triangleq \{(I,u) \text{ for all }I \in \mathcal{I} \mid t'(I) = u \}$. The function $t'$ defines a representative for every interval in $\mathcal{I}$ and hence $|S'| = m$. Further, since every node in $S'$ has distinct hyperedge co-ordinate, it follows that $S'$ is an exact hitting set of Type 1 cliques in $G_1$. 
We now show that $S'$ intersects every Type 2 clique at most $q'$ times. Let $G_1[S']$ be the subgraph of $G_1$ induced by nodes in $S'$. Since every maximal clique of Type 1 is hit exactly once by $S'$, it is sufficient to show that the size of a maximum clique in $G_1[S']$ is at most $q'$. In order to show this, we prove that if there is a clique $Q_1$ of size $q'$ in $G_1[S']$, then there is a clique $Q_2$ of size $q'$ in $\Gamma_{t'}$ with the following property. The vertices of $Q_2$ are exactly the vertex co-ordinates of nodes in $Q_1$. Formally, if $(I_1,u_1),(I_2,u_2),\ldots,(I_q,u_q)$ are nodes in $Q_1$, then there is a clique in $Q_2$ with the vertices $u_1,u_2,\ldots,u_q$. Observe that in $Q_1$, the vertex co-ordinates of every node will be distinct because by construction of $G_1$ there are no edges between nodes that have same vertex co-ordinates. Hence it is sufficient to show that if nodes $(I_i,u_i)$ and $(I_j,u_j)$ belong to $Q_1$, then there exists an edge between vertices $u_i$ and $u_j$ in $Q_2$. Since $S$ is an exact hitting set, the intervals $I_i$ and $I_j$ are not the same. Hence the only reason why the edge between $(I_i,u_i)$ and $(I_j,u_j)$ exists is because $u_i,u_j \in I_i$ or $u_i,u_j \in I_j$. Without loss of generality, let $u_i,u_j \in I_i$. Then, by the representative function $t'$, the representative of $I_i$ is the vertex $u_i$. Since $u_j$ also belongs to $I_i$, $(u_i,u_j)$ is an edge in $\Gamma_{t'}$. It follows that for every edge $e$ in $Q_1$, there exists a distinct edge $e'$ in $Q_2$. Hence there exists a clique of size $q'$ in $\Gamma_{t'}$ corresponding to a clique of size $q'$ in $G_1[S]$. Thus we conclude that if there is a CF colouring of $H$ using $q'$ colours, then there exists an exact hitting set of Type 1 cliques of $G_1$ that intersects every Type 2 maximal clique of $G_1$ at most $q'$ times.
\qed
\end{proof}
\subsection{Linear Program for Exact Hitting Sets of Type 1 Cliques} \label{sec:LP}
Given an interval hypergraph $H = (\mathcal{V},\mathcal{I})$, the linear program to find an exact hitting set of $\mathcal{Q}_1$.  From Lemma \ref{lem:hitsetequiv} we know that one such hitting set results in a representative function $t$ such that $\omega(\Gamma_t) = \chi_{min}(H)$. As described in the initial paragraph of Section \ref{sec:CFCfromEHSofType1Cliques}, the above goal translates into finding an exact hitting set of Type 1 cliques such that each clique in Type 2 is hit as few times as possible. 

In this LP, there is one variable corresponding to each node of $G_1$. Define $X \triangleq \{ x_{I,u} \mid (I,u) \in G_1\}$ to be the set of variables in the LP, where
\[
x_{I,u} = 
\begin{cases}
1, &\quad\text{if node } (I,u)  \text{ hits Type 1 clique corresponding to } I\\
0, &\quad\text{otherwise}
\end{cases} 
\]

\noindent
\textbf{LP Formulation. }

\settowidth{\LPlhbox}{(P.1)}%
\noindent%
\parbox{\LPlhbox}{\begin{align}
               \tag{P.1}
               \end{align}}%
\hspace*{\fill}%
\begin{minipage}{\linewidth}
 \begin{align}
&\text{Find values to variables }   \displaystyle \{x_{I,u} \mid u \in I, I \in \mathcal{I}\} \text{ subject to}  \nonumber\\
&\displaystyle\sum\limits_{u \in I} x_{I,u} = 1, \forall I \in 	\mathcal{I}  \\
&\displaystyle\sum\limits_{(I,u) \in Q} x_{I,u} \leq q,  \text{ for each maximal clique }Q \text{ in } \mathcal{Q}_2.\\
&x_{I,u} \leq 1 \nonumber
\end{align}
\end{minipage}
The LP has a set of equations, which are given in (P.1):(1) and a set of inequalities, which are given in (P.1):(2). Logically, an equation corresponds to choosing exactly one vertex per interval; that is, each equation corresponds to choosing exactly one node from one maximal clique in $\mathcal{Q}_1$. On the other hand, an inequality corresponds to a maximal clique in $\mathcal{Q}_2$. Logically, the inequality means that we pick at most $q$ nodes from every maximal clique in $\mathcal{Q}_2$. Together, the solution to the LP is an exact hitting set of maximal cliques in $\mathcal{Q}_1$ such that each maximal clique in $\mathcal{Q}_2$ is hit at most $q$ times. 

This LP is solved using the ellipsoid method which uses a polynomial time separation oracle that we next design. Note that the optimum solution thus obtained may have fractional values. Section \ref{subsubsec:RoundingLP} details a rounding technique that converts this fractional solution to a feasible integer solution for the LP in polynomial time. 
\subsection{Separation Oracle based LP Algorithm $\mathtt{SPAlg}$} \label{sec:SepOracle}
A separation oracle is a polyomial time algorithm that given a point in $\mathbb{R}^d$, where $d$ is the number of variables in a linear program relaxation, either confirms that this point is a feasible solution, or produces a violated constraint \cite{VaziraniApproxAlgo}. In this section, we describe a separation oracle \texttt{SPMaxWtClique} for our LP as follows. Recall that $X = \{ x_{I,u} \mid (I,u) \in V(G_1)\}$ is the set of variables in the LP formulation. Given an assignment $\phi: X \rightarrow \mathbb{Q}$, where $\mathbb{Q}$ is the set of rational numbers, \texttt{SPMaxWtClique} either confirms that $\phi$ is a feasible assignment or returns an infeasible inequality if the assignment is not feasible. We design the separation oracle for the interval hypergraph which has 3 disjoint intervals.  Consider the vertex-weighted graph $G_1^w$ corresponding to $G_1$, where the weight function $w: V(G_1^w) \rightarrow  \mathbb{Q}$ is defined as follows: $w\big((I,u)) = \phi(x_{I,u})$. Find the maximum weight clique of $G_1^w$. If the weight of the maximum weight clique of $G_1^w$ exceeds $q$, then it follows that there is some maximal clique $Q'$ whose weight is more than $q$. This implies that the given point violates  the inequality corresponding to $Q'$. If the weight of the maximum weight clique is at most $q$, then we check if all the equations are feasible. If some equation is violated, then again we have found a violated constraint.  This completes the description of the separation oracle \texttt{SPMaxWtClique}.  We show in Lemma \ref{lem:SepOracPoly} that \texttt{SPMaxWtClique} runs in polynomial time. 
\begin{lemma} \label{lem:SepOracPoly}
If the input interval hypergraph has at least 3 disjoint intervals, then the separation oracle \texttt{SPMaxWtClique} runs in polynomial time.
\end{lemma}
\begin{proof}
If the input interval hypergraph has at least 3 disjoint intervals, then the vertex-weighted graph $G_1^w$ is perfect by Theorem \ref{thm:ResConfGraphIsPerfect}. It is known from \cite{Grotschel1981} that the maximum weight clique problem in perfect graphs can be solved in polynomial time. 
Thus, finding an inequality in the LP corresponding to a maximal clique whose weight exceeds $q$ can  be done in polynomial time. Also, since there are only a polynomial number of Type 1 clique, it follows that the check of whether there is a violated equation can also be done in polynomial time.  It follows that \texttt{SPMaxWtClique} runs in polynomial time.
\qed
\end{proof}
Let $\mathcal{B}$ be an instance of the given LP.  
Now, we show that the LP can be solved in polynomial time.
\begin{lemma} \label{lem:SPAlgPoly}
If there are at least 3 disjoint intervals in the input interval hypergraph, then the algorithm $\mathtt{SPAlg}$ runs in polynomial time.
\end{lemma}
\begin{proof}
We have shown in Lemma \ref{lem:SepOracPoly} that the separation oracle in $\mathtt{SPAlg}$ runs in polynomial time when there are at least 3 disjoint intervals in $\mathcal{I}$. Since there is a polynomial time separation oracle, by referring to the ellipsoid method, optimization in the polytope of $\mathcal{B}$ can be done in polynomial time.
\qed
\end{proof}
We now describe Algorithm $\mathtt{SPAlg}$.   
Algorithm $\mathtt{SPAlg}$ takes as inputs the LP instance $\mathcal{B}$ and an integer $q \geq 1$. It uses the separation oracle \texttt{SPMaxWtClique} and returns an assignment of values to variables in $X$ if the system is feasible.  Otherwise, it reports that the system is infeasible. Let $q_{min}$ be the smallest value of $q$ for which Algorithm $\mathtt{SPAlg}$ finds a feasible solution of the instance $\mathcal{B}$ and let $B_{opt}$ be the solution returned by Algorithm $\mathtt{SPAlg}$. If $B_{opt}$ is an integral solution, then we have an integer solution in polynomial time. If $B_{opt}$ is not integral, then we present steps to round the fractional values in $B_{opt}$ that results in a feasible integral solution for the value $q_{min}$.
\subsection{Rounding the LP solution}\label{subsubsec:RoundingLP} \texttt{RoundingAlgo} described in Algorithm \ref{algo:RoundingAlgo} takes as input a fractional feasible solution of the LP $\mathcal{B}$ and the integer $q_{min}$ and returns a feasible integer solution for $\mathcal{B}$ for the value $q_{min}$. 

\begin{algorithm}[H]
\caption{\texttt{RoundingAlgo}}
\label{algo:RoundingAlgo}
\KwIn{$B_{opt},\mathcal{I}'$}
\vspace{3mm}
$i \leftarrow 0$ \;
$B_{opt}(0) \leftarrow B_{opt}$ \;
\While {$\exists x_{I,v} \in B_{opt}(i)$ that does not belong to $\{0,1\}$}{ 
$i \leftarrow i +1$ \;
$B_{opt}(i) \leftarrow B_{opt}(i-1)$ \;
$I_i \leftarrow $ Longest Interval in $\mathcal{I}'$ with the smallest left endpoint \;
$r \leftarrow r(I_i)$ \;
$r-1 \leftarrow $ vertex to the immediate left of $r(I_i)$ on the line \;
\For{each interval $I'$ that contains $r$ and $r-1$}{
$x_{I',r-1} \leftarrow x_{I',r-1} + x_{I_i,r}$ \; \label{algLine:rounding1}
$x_{I',r} \leftarrow x_{I',r} - x_{I_i,r}$ \;
Modify entries in $B_{opt}(i)$ corresponding to the values changed above \;
\If {$x_{I',r} = 0$}{
$\mathcal{I}' = \mathcal{I}' \setminus I' \cup (I' \setminus r)$ \;
}
}
}
$B_{optI} \leftarrow B_{opt}(i)$ \;
return $B_{opt}(i)$ \;
\end{algorithm}
In every iteration of the \textit{while} loop in Algorithm \ref{algo:RoundingAlgo}, at least one variable in $X$ is rounded to an integer value. 
In iteration $i$, let $I_i$ be the interval with the smallest left end point among all intervals of maximum length. Let $l(I_i)$ and $r(I_i)$ denote the left and right endpoints of interval $I_i$ respectively. Since $r(I_i)$ is removed during iteration $i$, it follows that the total number of points (in all the intervals) in iteration $i+1$ is at least one less than the total number of points in iteration $i$. Hence the conflict graph corresponding to intervals in iteration $i+1$ has strictly fewer number of nodes than the conflict graph corresponding to intervals in iteration $i$. 
In Lemma \ref{lem:feasibleAfterRounding}, we show that for every $i \geq 0$, the solution $B_{opt}(i)$ is feasible for the linear program $\mathcal{B}$ for the value $q_{min}$.   
We show in Lemma \ref{lem:RoundResultsIntegerSoln} that for some positive integer $j$, $B_{opt}(j)$ will be an all integer solution for $\mathcal{B}$, at which time algorithm exits.
\begin{lemma} \label{lem:RoundResultsIntegerSoln}
Let $B_{opt}$ be a fractional feasible solution returned by $\mathtt{SPAlg}(\mathcal{B}, q_{min})$. Then, \texttt{RoundingAlgo} returns an integer solution for $\mathcal{B}$ on the input $B_{opt}$ in a polynomial number of steps.
\end{lemma}
\begin{proof}
From the description of \texttt{RoundingAlgo}, in each iteration $i$, $x_{I_i,r(I_i)}$ becomes zero and the variable $x_{I_i,r(I_i)}$ does not become non-zero in any subsequent iteration.   Then the number of variables whose value is not 0 or 1 reduces in each iteration.  Further,  the rounding is such that if a variable $x_{I,r}$ is reduced by a certain value then $x_{I,r-1}$ is increased by the exact same value.  This ensures that after each iteration the equations the equations  in (P.1):(1)
are all satisfied, and in particular they add up to 1.  Therefore, eventually in each equation there will be a variable which is 1 and all others are 0.   It follows that the solution will be integral in at most $\mu(H)$ iterations, where $\mu(H)$ is the number of nodes in $G_1$.
\qed
\end{proof}

Let $B_{optI}$ be the integer solution returned by \texttt{RoundingAlgo}. We show in Lemma \ref{lem:feasibleAfterRounding} that $B_{optI}$ is feasible for the instance $\mathcal{B}$ for the value $q_{min}$.  That is the values to the variables in each Type 2 inequality add up to at most the same value as it was adding up to in $B_{opt}$.  The proof of correctness for the rounding algorithm is by induction on the number of nodes in $G_1$. We show that the solution returned on a smaller instance after every iteration is feasible for $\mathcal{B}$. In the proof of Lemma \ref{lem:feasibleAfterRounding} below, we use $r$ to denote $r(I_i)$, where $I_i$ is the longest interval with the smallest left endpoint in iteration $i$. Similarly, denote the point to the immediate left of $r$ on the number line by $r-1$. For every other interval $I'$, denote its right endpoint and the point immediately to the left of the right endpoint by $r(I')$ and $r(I') - 1$ respectively. 

\begin{lemma} \label{lem:feasibleAfterRounding}
Let $B_{opt}$ be a fractional feasible solution returned by $\mathtt{SPAlg}(\mathcal{B}, q_{min})$. 
The solution $B_{optI}$ returned by \texttt{RoundingAlgo} is a feasible solution for the LP instance $\mathcal{B}$ for the value $q_{min}$.
\end{lemma}
\begin{proof}
The proof of correctness is by induction on the iteration number.   We know that $B_{opt}$ is feasible for $\mathcal{B}$.  Let us assume that for an integer $i \geq 0$
 $B_{opt}(i-1)$ is feasible for $\mathcal{B}$. We show that $B_{opt}(i)$ is also feasible for $\mathcal{B}$. 
From the description of the \texttt{RoundingAlgo}, during iteration $i$, the value which is subtracted from one variable from $x_{I,r}$ is added to the variable $x_{I,r-1}$. This fact is crucially in the analysis below.   Hence all equations in (P.1):(1) are satisfied by $B_{opt}(i)$. Now, we show that the inequalities in (P.1):(2) corresponding to the maximal cliques are also satisfied by $B_{opt}(i)$.  Let $I'$ be an interval that contains the point $r-1$ such that $x_{I',r-1}$ has increased due to step \ref{algLine:rounding1} in Algorithm \ref{algo:RoundingAlgo}. By the choice of $I'$ for which $x_{I',r-1}$ is increased, it follows that $x_{I',r}$ is reduced and thus $I'$ contains the point $r$. 
It follows from the definition  of the edge set  $E_{colour}$ that there is an edge between $(I',r-1)$ and $(I_i,r)$ in $G_1$. 

Let $Q$ be a maximal clique that contains the node $(I',r-1)$. By Observation \ref{obs:NodesOfaVertexAreIndep}, all nodes with the same vertex co-ordinate form an independent set. Hence $Q$ does not contain any node of the form $(I'',r-1)$, where $I'' \neq I'$. 
If $Q$ contains the node $(I',r)$, then $x_{I',r}$ has reduced and hence the inequality corresponding to $Q$ is satisfied under $B_{opt}(i)$. 
If $Q$ does not contain the node $(I',r)$, then among all nodes in $Q$, consider two nodes - one for which the vertex coordinate is leftmost and another for which the vertex coordinate is the rightmost on the line. We denote the leftmost coordinate by $\lambda$ and the rightmost coordinate by $\rho$. Let $(J,\lambda)$ and $(J',\rho)$ be  two nodes in $Q$. 

First, we show that $\lambda \geq l(I_i)$. The proof is by contradiction. Suppose $\lambda < l(I_i)$. Due to the edge between nodes $(J,\lambda)$ and $(I',r-1)$ in $Q$, it is clear that either $J$ or $I'$ contains both $\lambda$ and $r-1$. Without loss of generality, assume that $J$ contains both $\lambda$ and $r-1$. Since by our assumption $\lambda < l(I_i)$, it follows that $J$ is at least as long as $I_i$ and $l(J) < l(I_i)$. This is a contradiction to our choice of $I_i$ being the longest interval with the smallest left endpoint. It follows that $\lambda \geq l(I_i)$. 
We show using the following cases that the inequality corresponding to $Q$ is still feasible.
\begin{enumerate} 
\item Case $\rho < r-1$. We show that this case is not possible. Since $(I',r-1)$ belongs to $Q$, and $\rho$ is the rightmost vertex co-ordinate among all nodes in $Q$, it follows that $\rho \geq r-1$. 
\item Case $\rho = r-1$. 
Since $\lambda \geq l(I_i)$ and $\rho = r-1$, it follows that all points from $\lambda$ to $\rho$ belong to $I_i$.  Therefore, by the definition of the edges of $G_1$, $(I_i, r)$ is adjacent to all the nodes of $Q$ whose vertex coordinates are between $\lambda$ and $\rho$, both included.  This contradicts the premise that $Q$ is a maximal clique. Therefore $\rho = r-1$ is not possible.  
\item Case $\rho = r$. Since $(J',\rho)$, which is the same as $(J',r)$ belongs to $Q$, it follows that the inequality corresponding to $Q$ is still feasible.  Since the decrease in $x_\{J',r\}$ is exactly the same as the increase in $x_{I_i,r-1}$.  
\item Case $\rho > r$. Observe that there is an edge between nodes $(J,\lambda)$ and $(J',\rho)$ since they are both in $Q$. It follows that either $J$ or $J'$ both contain $\lambda$ and $\rho$. Without loss of generality, let $J$ be this interval. Since $J$ contains all the points on the line from $\lambda$ to $\rho$, both included, it follows that the interval $J$ contains both points $r$ and $r-1$. Further, by the definition of the graph $G_1$, it follows that $(J,r)$ is adjacent to all the nodes in $Q$ whose vertex coordinates lie between $\lambda$ and $\rho$, both included.  Further, since there can be at most one node in a maximal clique with a vertex coordinate, and since $Q$ is a maximal clique, it follows that  $(J,r)$ belongs to $Q$.  Since $J$ also contains the point $r-1$, and since $x_{I_i,r}$ is reduced in iteration $i$, follows that $x_{J,r}$ is also reduced and $x_{J,r-1}$ is increased in iteration $i$.  Therefore, in the maximal clique $Q$ the increase in $x_{I',r-1}$ is compensated by a decrease in $x_{J,r}$.  Therefore, the inequality corresponding to $Q$ is satisfied in by $B_{opt}(i)$.  
\end{enumerate}
Therefore, in all the cases we have concluded the $B_{opt}(i)$ satsfies ${\mathcal B}$.  This completes the proof by induction.
\qed
\end{proof}

\noindent
We show in Theorem \ref{thm:CFCIntHypPolyTime} that the CF colouring problem in interval hypergraphs can be solved in polynomial time by considering two cases of $H$ - when there are at least 3 disjoint intervals in $H$ and when there are at most two disjoint intervals in $H$. In Lemma \ref{lem:3DisjIntCFCPolyTime}, we show that the first case can be solved in polynomial time and in Lemma \ref{lem:2disjInt2CFC}, we show  a polynomial time solution  for the second case.
\begin{lemma} \label{lem:3DisjIntCFCPolyTime}
Let $H = (\mathcal{V},\mathcal{I})$ be an interval hypergraph such that there are at least 3 disjoint intervals in $\mathcal{I}$. Then, the CF colouring problem in $H$ can be solved in polynomial time.
\end{lemma}
\begin{proof}
By Lemma \ref{lem:SPAlgPoly}, when there are at least 3 disjoint intervals in $\mathcal{I}$, the LP returns a feasible solution in polynomial time using the separation oracle \\ \texttt{SPMaxWtClique}. By Lemmas \ref{lem:RoundResultsIntegerSoln} and \ref{lem:feasibleAfterRounding}, a feasible integer solution can be obtained from the fractional feasible solution in polynomial time. Further, the representative function $t$ and thereof, the co-occurrence graph $\Gamma_t$ can also be obtained in polynomial time. By Theorem \ref{thm:Co-occPerf}, the co-occurrence graph $\Gamma_t$ is perfect. Since a proper colouring of a perfect graph can be found in polynomial time, it follows from Theorem \ref{thm:Co-occChar} that an optimal CF colouring of an interval hypergraph can be found in polynomial time.
\qed
\end{proof}

We next bound the CF colouring number of an interval hypergraph which does not contain  3 pairwise disjoint intervals.
\begin{lemma} \label{lem:2disjInt2CFC}
Let $H = (\mathcal{V},\mathcal{I})$ be an interval hypergraph which does not contain three disjoint intervals in $\mathcal{I}$.  Then $\chi_{cf}(H) \leq 2$. Further, such an interval hypergraph can be recognized in polynomial time.
\end{lemma}
\begin{proof}
Let us consider the intersection graph of the set of intervals $\mathcal{I}$ which we know is an interval graph. It is well-known (see for example the book by Golumbic \cite{Gol2004}) that the interval graph is perfect. Since there do not exist 3 disjoint intervals, it follows that the interval graph has a maximum independent set of size at most 2.  From the definition of perfect graphs (see Section \ref{sec:Prelims}), we know that the size of the maximum independent set is equal to the size of the minimum clique cover, and it can be found in polynomial time.  Therefore, the interval graph of $\mathcal{I}$ has a clique cover of size at most 2.  The clique cover gives at most two corresponding points in $\mathcal{V}$ that  intersect each interval in $\mathcal{I}$.   We now consider the following vertex colouring function defined on $\mathcal{V}$:
colour one of the points with colour 1 and the other point, if necessary, with colour 2, and all the other points are coloured 0.  Since at most colours 1 and 2 are given to at most two vertices in $\mathcal{V}$ and all the other vertices are given the colour 0, this vertex colouring function is a CF colouring of $H$.  Thus, $\chi_{cf}(H) \leq 2$ and the recognition of such interval hypergraphs can also be done in polynomial time. 
\qed
\end{proof}

\noindent
Finally, we prove the main result in this paper. 
\begin{proof}[of Theorem \ref{thm:CFCIntHypPolyTime}]
If $H$ is an exactly hittable interval hypergraph then, by Lemma \ref{lem:1CFimpliesEHS}, $\chi_{cf}(H) = 1$.  From Theorem \ref{thm:IHisEH}, an exactly hittable interval hypergraph can be recognized in polynomial time. If $H$ is not exactly hittable and if $H$ does not have 3 disjoint intervals, then by Lemma \ref{lem:2disjInt2CFC}, $\chi_{cf}(H) \leq 2$ and such an $H$ can be recognized in polynomial time. 
If the minimum clique cover of $H$ is  at least $3$, then there are at least three pairwise disjoint intervals in $\mathcal{I}$. It has been shown in Lemma \ref{lem:3DisjIntCFCPolyTime} that the CF colouring problem can be optimally solved in polynomial time. From the above results it follows that an optimal CF colouring of $H$ can be obtained in polynomial time.
\qed
\end{proof}

\section{Partition into Exactly Hittable Sets and Conflict-free Colouring} \label{sec:cfeqehs}
Using Lemmas \ref{lem:CFimpliesEHS}, \ref{lem:kPartskClique} and \ref{lem:kClique} we prove Theorem \ref{thm:EHS-CF}.    We first show that for any hypergraph a CF colouring with $k$ colours gives a partition of the hyperedges into exactly $k$ exactly hittable hypergraphs.  Then we prove interval hypergraphs that can be partitioned into $k$ exactly hittable hypergraphs can be CF coloured with $k$ colours.
\begin{lemma} \label{lem:CFimpliesEHS}
If there exists a CF colouring of a hypergraph $H = (\mathcal{V},\mathcal{E})$ with $k$ non-zero colours, then there exists a partition of $\mathcal{E}$ into $k$ parts $\{\mathcal{E}_1, \ldots, \mathcal{E}_k\}$ such that each $H_i = (\mathcal{V}, \mathcal{E}_i), 1 \leq i \leq k$ is an exactly hittable hypergraph.  
\end{lemma}
\begin{proof}
Given a CF colouring $C$ with at most $k$ non-zero colours, let $t$ be a representative function $t:\mathcal{E} \rightarrow \mathcal{V}$ such that for each $e \in \mathcal{E}$, $e$ is CF coloured by the vertex  $t(e)$.  The hyperedges are partitioned into sets $\{\mathcal{E}_1, \ldots \mathcal{E}_k\}$ based on $t$ and the vertex colouring as follows:  the set $\mathcal{E}_i$ consists of all those hyperedges $e \in \mathcal{E}$ such that the colour of  $t(e)$ is $i$.  We show that  for each $1 \leq i \leq k$, $H_i=(\mathcal{V},\mathcal{E}_i)$ is an exactly hittable hypergraph and the exact hitting set is $h_i = \{t(e) \mid e \in \mathcal{E}_i\}$.   
$h_i$ is a hitting set of $\mathcal{E}_i$ because for each $e \in \mathcal{E}_i$, $t(e)$ is in $h_i$. Since all the vertices of $h_i$ have the same colour assigned by $C$, it follows that each $e \in \mathcal{E}_i$ is hit exactly once by $h_i$. Thus, $h_i$ is an exact hitting set of $\mathcal{E}_i$.  Therefore, each $H_i = (\mathcal{V},\mathcal{E}_i)$ is an exactly hittable hypergraph. This proves the  lemma.\qed
\end{proof}We next set up the machinery to conclude that if we are given a partition of an interval hypergraph $H$ into $k$ exactly hittable interval hypergraphs, then we get a CF colouring with at most $k$ non-zero colours.
Let $P = \{\mathcal{E}_1, \mathcal{E}_2, \ldots, \mathcal{E}_k\}$ be a partition of intervals in $\mathcal{E}(H)$, such that each $H_i = (\mathcal{V},\mathcal{E}_i), 1 \leq i \leq k$  is an exactly hittable interval hypergraph. We show that there is a CF colouring of $H$ with $k$ non-zero colours. Let $h_1, \ldots, h_k$ be the exact hitting sets of the parts $\mathcal{E}_1, \ldots, \mathcal{E}_k$ respectively. Let $R$ denote the set $\cup_{i=1}^k  h_i$. For each interval $I \in \mathcal{E}_i$, let  $t(I)$ be the only vertex in $I \cap h_i$.  
Let $\Gamma_t$ be the co-occurrence graph of $H$.  In the arguments below, the graph $\Gamma_t$ and the representative function $t$ are as defined here.   We now prove Lemmas \ref{lem:kPartskClique} and \ref{lem:kClique} and use them in the proof of Theorem \ref{thm:EHS-CF}.  
\begin{lemma} \label{lem:kPartskClique}
Let $Q$ = $\{u_1, \ldots, u_q\}$ be a clique of size $q$ in the co-occurrence graph $\Gamma_t$. Then, there are $q$ distinct parts $s_1, \ldots, s_q$ in the set $\mathcal{P}$ containing intervals $I_1, \ldots, I_q$ respectively, satisfying the following property: for each $u_i$ in $Q$, $u_i$ is the representative of interval $I_i$ and for each edge $(u_i,u_j)$ in $Q$ either $u_j$ is in $I_i$ or $u_i$ is in $I_j$.
\end{lemma}
\begin{proof}
The proof is by induction on the size $q$ of the clique. The claim is true for base case when $q = 1$; then $u_1$ is the representative of some interval $I_1$ in some part $s_1$. Assume that the claim is true for any clique of size $q-1$. Now, we show that the claim is true for clique $Q$ of size $q$. Let $u_1 < \ldots <u_q$ be the left to right ordering of the points (on the line) corresponding to vertices in the clique $Q$. Since $(u_1,u_q)$ is an edge in $Q$, there must exist an interval $I$ such that either $u_1$ or $u_q$ is the representative of $I$ and $u_q$ occurs along with $u_1$ inside $I$. Without loss of generality, assume that $u_1$ is the representative of interval $I$. Observe that $I$ must contain all points in $u_1, \ldots, u_q$. By the induction hypothesis for the points $u_2, \ldots, u_q$, there are parts $s_2, \ldots, s_q$ and intervals $i_2,\ldots,i_q$ such that $u_i$ is representative of $I_i$ and for each edge $(u_i,u_j)$ in the clique on points in $u_2, \ldots, u_q$, either $u_j$ is in $I_i$ or $u_i$ is in $I_j$. We now show that $I$ does not belong to the parts $s_2, \ldots, s_q$ and that it belongs to a different part. Assume for contradiction that $I$ belongs to some part $s_j$ in $\{s_2, \ldots, s_q\}$. Then, the exact hitting set of set $s_j$ contains at least one point $u_j \in Q$ that is distinct from $u_1$. $u_j$ and $u_1$ cannot be the same point because there is an interval $I_j$ in $s_j$ whose representative is $u_j$. Observe that $u_1$ which is the representative of $I$ must also be in the exact hitting set of $s_j$ because according to our assumption, $I$ belongs to $s_j$. Since interval $I$ contains all points in $u_1, \ldots, u_q$, it is hit at least twice by the exact hitting set of $s_j$ which is a contradiction. \qed
\end{proof}

\begin{lemma} \label{lem:kClique}
The clique number of the co-occurrence graph $\Gamma_t$ is at most $k$.
\end{lemma}
\begin{proof}
For a clique $K$ of $q$ vertices in $\Gamma_t$, we know from  Lemma \ref{lem:kPartskClique} that there must be $q$ distinct exactly hittable parts $s_1, \ldots s_{q}$ and $q$ intervals $I_1, \ldots I_{q}$ in each part, respectively,  satisfying an additional property which is not important for this argument.  
Therefore, the size of the largest clique in $\Gamma_t$ is at most the number of parts which is at most $k$. \qed
\end{proof}

\begin{proof}[of Theorem \ref{thm:EHS-CF}]
From  Lemma \ref{lem:CFimpliesEHS}, it follows that if there is a CF colouring of a hypergraph $H$ with at most $k$ non-zero colours, then there is a partition of $H$ into at most $k$ exactly hittable hypergraphs.  To prove the other direction, given a partition of interval hypergraph $H$ into $k$ exactly hittable interval hypergraphs, we consider $\Gamma_t$ as defined before Lemma \ref{lem:kPartskClique}.   From Lemma \ref{lem:kClique}, the clique number of $\Gamma_t$ is at most $k$.  From Theorem \ref{thm:Co-occPerf}, we know that $\Gamma_t$ is a perfect graph. By the Perfect Graph Theorem \cite{Gol2004}, $\chi(\Gamma_t) = \omega(\Gamma_t) \leq k$.  Further from Theorem \ref{thm:Co-occChar}, $\chi_{cf}(H) \leq \chi(\Gamma_t) \leq k$.
 Thus, if there exists a partition of interval hypergraph $H$ into $k$ exactly hittable interval hypergraphs, then there exists a CF colouring of $H$ using at most $k$ non-zero colours.  Hence Theorem \ref{thm:EHS-CF} is proved. \qed
\end{proof}
{\bf Conclusion:}  Our algorithm also gives a $k$-Strong Conflict Free Colouring for each $k \geq 1$ for interval hypergraphs which have at least 3 disjoint intervals.  This is achieved by following the approach of writing an LP to hit Type 1 cliques in $G_1$ at least $k$ times while hitting each Type 2 clique as few times as possible.  The LP has separation oracle when $G_1$ is perfect which we know for sure when the interval hypergraph has 3 disjoint intervals.  The case when there are at most two disjoint intervals is  a direction of future work for $k$-SCF colouring.  
\bibliographystyle{plain}
\bibliography{cfc}
\end{document}